\documentclass[10pt, journal]{IEEEtran}
\usepackage{cite}

\ifCLASSINFOpdf
   \usepackage[pdftex]{graphicx}
\else
\fi
\usepackage{amsmath}
\usepackage{algorithmic}

\usepackage{array}
\usepackage{xcolor}
\usepackage[ruled,linesnumbered]{algorithm2e}
\usepackage{amsthm}
\newtheorem{property}{Property}
\newtheorem{lemma}{Lemma}
\newtheorem{theorem}{Theorem}

\usepackage{multirow}
\usepackage{booktabs}
\usepackage{balance}
\usepackage{color}
\usepackage{hyperref}

\begin{document}
%
\title{A Switch Architecture for Time-Triggered Transmission with Best-Effort Delivery}
%
%
%

\author{Zonghui~Li,
        Wenlin~Zhu,
        Kang~G.~Shin,~\IEEEmembership{Life~Fellow,~IEEE,}
        Hai~Wan,
        Xiaoyu~Song,~\IEEEmembership{Senior~Member,~IEEE,}
        Dong~Yang,~\IEEEmembership{Senior~Member,~IEEE,}
        and~Bo~Ai,~\IEEEmembership{Fellow,~IEEE}
\IEEEcompsocitemizethanks{
\IEEEcompsocthanksitem Zonghui Li, Wenlin Zhu are with the Beijing Key Laboratory of Transportation Data Analysis and Mining, School of Computer and Information Technology, 
Beijing Jiaotong University, Beijing, China, 100044.
Zonghui Li is the corresponding author. E-mail: zonghui.lee@gmail.com, zhuwenlin@bjtu.edu.cn.
\IEEEcompsocthanksitem Kang G. Shin is with the Department of 
Electrical Engineering and Computer Science, University of Michigan, Ann Arbor, MI 48109-2121, USA. E-mail: kgshin@umich.edu. 
\IEEEcompsocthanksitem Hai Wan is with the Software School, 
Tsinghua University, Beijing, China, 100084. E-mail: 
wanhai@tsinghua.edu.cn. 
\IEEEcompsocthanksitem Xiaoyu Song is in the Department of Electrical and Computer Engineering, Portland State University, Portland, OR. E-mail: songx@pdx.edu.
\IEEEcompsocthanksitem Dong Yang, Bo Ai are with the School of Electronic and Information Engineering, Beijing Jiaotong University, Beijing, China, 100044. E-mail: dyang@bjtu.edu.cn, boai@bjtu.edu.cn.
}
}

\maketitle

\begin{abstract}
In Time-Triggered (TT) or time-sensitive networks, the 
transmission of a TT frame is required to be scheduled at 
a precise time instant for industrial distributed 
real-time control systems. Other (or {\em best-effort} (BE)) 
frames are forwarded in a BE manner. 
Under this scheduling strategy, the transmission of a TT frame must wait 
until its scheduled instant even if it could have been transmitted sooner.
On the other hand, BE frames are transmitted whenever possible 
but may miss deadlines or may even be dropped due to congestion. 
As a result, TT transmission and BE delivery are
incompatible with each other.

To remedy this incompatibility, we propose a synergistic 
switch architecture (SWA) for TT transmission with BE delivery to 
dynamically improve the end-to-end (e2e) latency of TT frames 
by opportunistically exploiting BE delivery. 
Given a TT frame, the SWA generates and transmits
a cloned copy with BE delivery. The first frame arriving at 
the receiver device is delivered with a configured 
jitter and the other copy ignored. So, the SWA
achieves shorter latency and controllable jitter,
the best of both worlds. 
We have implemented SWA using FPGAs in an
industry-strength TT switches and used four test scenarios to 
demonstrate SWA's improvements of e2e latency and controllable 
jitter over the state-of-the-art TT transmission scheme.
\end{abstract}


%
\IEEEpeerreviewmaketitle

\section{Introduction}
%
%
%
%
\IEEEPARstart{I}{n} the past decade, Ethernet has been increasingly used/deployed 
in industrial distributed real-time control systems \cite{Vitturi2019Industrial, Gavrilut2022Constructive}.
However, the conventional Ethernet delivers frames in a best-effort (BE) manner 
without accounting for time-critical properties for industrial control tasks. 
Industrial Ethernet \cite{LoBello2019Perspective} is a set of new solutions that 
integrate industrial control networks and standard Ethernet. 
The time-triggered (TT) communication paradigm \cite{Bruckner2019Introduction} is 
a cost-efficient solution for Industrial Ethernet. Time-critical frames for 
industrial control functions, called {\em TT frames}, are statically scheduled 
for transmission at precise time instants\cite{Minaeva2021Survey}. Meanwhile, the other (called {\em BE})
frames are forwarded with BE delivery of the standard Ethernet. We call such 
networks with a TT communication paradigm {\em TT networks} \cite{Li2019enhanced}.

Two typical networks for Industrial Ethernet fall into the category of TT networks. 
One is {\em TTEthernet}, standardized as AS6802 \cite{saeAS6802} targeting the 
aerospace domain by the Society of Automotive Engineers (SAE) International Group. 
It defines a fault-tolerant synchronization strategy to build 
and maintain synchronized time in a distributed system 
including terminal systems and switches. 
TT schedulers \cite{SteinerSMTScheduler, pozoSMTScheduler15p, 
pozoSMTScheduler15, craciunas2016combined} are used to 
statically schedule TT frames according to the 
application-specific requirements for the worst-case e2e latency. 
The other is {\em Time-Sensitive Networking} (TSN) \cite{Finn2018Introduction}, 
defined and refined by the IEEE 802.1 TSN Task Group as an extension of the 
IEEE 802.3 Ethernet for real-time transmission since 2012. It employs IEEE 802.1AS 
\cite{IEEE8021AS} (based on IEEE 1588 \cite{IEEE1588}, a precision clock synchronization 
protocol) to synchronize all devices participating in real-time communication. 
Traffic is labeled with different priorities and scheduled to meet different 
time-critical requirements in the mixed transmission, which is standardized as 
IEEE 802.1Qbv \cite{IEEE8021Qbv}. 802.1Qbv uses time-aware transmission gates to 
separate transmission queues for different traffic classes. The transmission gates 
are opened and closed at specific times according to a time-based circular scheduler. 
These TT schedulers \cite{SteinerSMTScheduler, pozoSMTScheduler15p, pozoSMTScheduler15, 
craciunas2016combined} are compatible with the time-aware transmission and have 
already been applied to TSN in \cite{CraciunasTSNScheduling, SernaOliver2018IEEE, Steiner2018Traffic}.

However, scheduling TT frames by assigning them precise transmission time instants 
is a bin-packing problem, which is known to be NP-complete. 
Thus, TT schedulers \cite{SteinerSMTScheduler, pozoSMTScheduler15p, pozoSMTScheduler15, 
craciunas2016combined, CraciunasTSNScheduling, Steiner2018Traffic, SernaOliver2018IEEE} 
usually search for application-specific worst-case e2e latency constraints as 
deadlines so as to reduce the time to solve and enhance schedulability. 
Furthermore, under this scheduling policy, the TT frames cannot be sent sooner 
than their scheduled sending time even when the network is idle.
For example, given a frame of 64 bytes under fast Ethernet (100 Mbps), the 
application-specified e2e latency 200$\mu s$, and the path from A to F illustrated 
in Fig.~\ref{fig:TTNetworkExample}, a feasible solution for the scheduled sending 
times may be 0$ns$ for end-device A, 80,000$ns$ for switch D, and 
180,000$ns$ for switch E. 
This solution meets the latency requirement but needs about $\approx$180$\mu s$ 
to deliver the frame. In fact, using BE delivery, the forwarding latency for 
the frame, defined as the latency from the first bit received by the switch to 
the first bit sent out by the switch, is 7.92$\mu s$ in the absence of congestion 
in our fast Ethernet switch. Therefore, the e2e latency is about 20$\mu s$ 
from A to F without congestion. 
The latency is 10x shorter than that of TT transmission. 
As a result, TT frames are delivered much slower than non-real-time BE frames! 
That is, there exists a big latency gap between TT transmission and BE transmission, 
making us wonder whether it is good enough for time-critical applications 
to only guarantee their deadlines by the TT transmission mode.

In the context of Industry 4.0 \cite{Wollschlaeger2017Future, LoBello2019Perspective}, 
the timely delivery of control frames by industrial networks is vital
to different control steps in industrial distributed systems. 
The faster the frames are delivered, the closer these control steps are in order 
to achieve shorter response time and higher efficiency. 
For example, fast delivery for emergency braking in high-speed trains can gain more time 
to avoid accidents. Another example is healthcare monitoring, especially for life-critical 
heart condition. Besides meeting deadlines, fast delivery of such monitoring messages 
will enable a more timely and efficient response to health problems. 
So, networks need to reduce delay beyond what e2e deadlines require. 
However, TT transmission satisfies the e2e deadlines of time-critical frames 
by scheduling their precise sending instants in switches that could hamper their fast delivery. 
On the other hand, BE frames are forwarded using strategies like strict priority (SP), 
weighted round robin priority (WRR), credit-based traffic shaping (CBS)\cite{IEEE8021Qav, Zhao2018Timing, zhao2021Latency} and asynchronous traffic shaping (ATS)\cite{IEEE8021Qcr, zhao2022Quantitative}, so as to ensure 
their delivery as soon as possible. Nonetheless, such transmission strategies do not 
guarantee the e2e deadlines and may even lead to their loss due to queuing, rerouting, congestion, etc. 
So, the BE transmission is unacceptable for industrial control applications, and  
moreover, TT transmission and BE delivery are incompatible to each other.

To exploit both BE and TT transmissions, we propose a synergistic 
switch architecture (SWA) to forward TT frames as soon as possible 
--- without waiting until their scheduled sending times --- and 
meet their e2e latency constraints; our preliminary results were 
reported in \cite{Li2019DOE}. The architecture forwards the cloned 
TT frames with BE delivery to speed up the transmission of TT frames. 
That is, a TT frame and its copy are forwarded with TT and BE 
strategies, respectively. Whichever of the two copies arrives at 
the end-device first is delivered to the receiver application
and the other copy is discarded. 
As a result, the e2e latency of the delivered copy is likely to be 
shorter than that of the TT frame. Furthermore, to handle the possible 
loss of frame copies, cloning happens at each switch along the routing 
path of a TT frame. By comparing the strictly-increasing sequence number
of TT frames, the SWA guarantees the transmission of only one 
cloned copy of each TT frame. Therefore, the bandwidth cost for 
the BE transmission of cloned copies is equal to the cost 
bandwidth of TT frames.

\begin{figure}[htbp]
  \centering
  \includegraphics[width=3.2in]{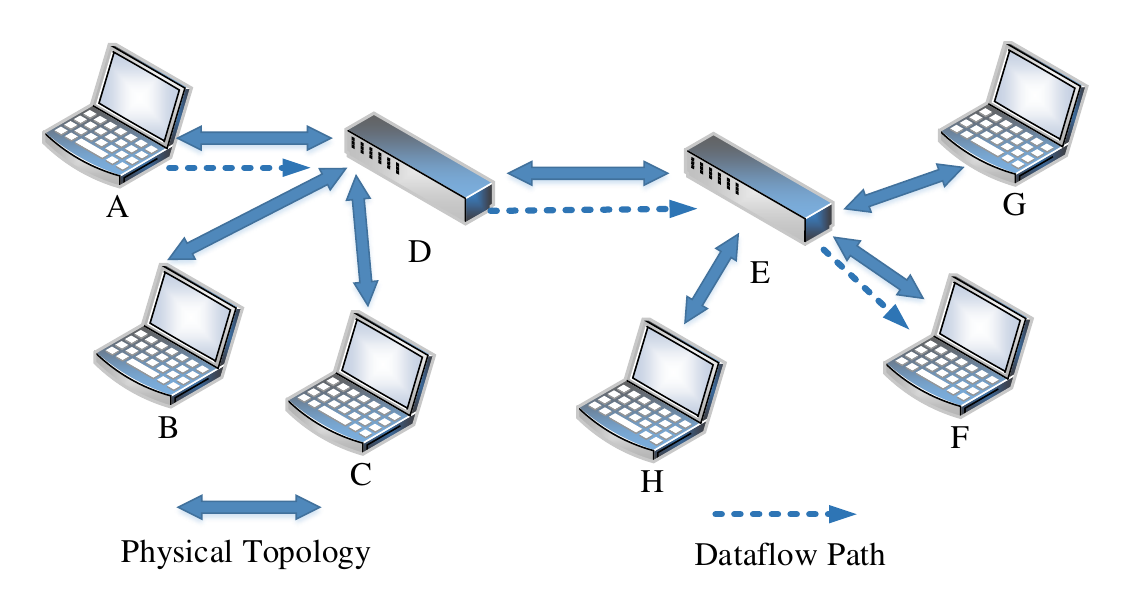}
  \caption{Example of a TT network with 6 terminal systems and 2 switches.}
  \label{fig:TTNetworkExample}
\end{figure}

In fact, the SWA only uses the unused bandwidth of 
BE transmission because it opportunistically delivers these copies. 
The congestion, if happens, lets the copies be dropped. 
Once the copies have an opportunity to be transmitted, they will 
be delivered again. Such an opportunism allows the SWA to 
dynamically improve the e2e latency of TT frames and adapts to 
the changes of the available BE bandwidth. 
Moreover, the bandwidth cost of TT frames for industrial control
systems is usually low. For example, the multifunction 
vehicle bus (MVB) \cite{IEC613751} for a train communication network 
has a bandwidth of up to 1.5 Mbps. The controller area network (CAN) 
\cite{ISO11898CAN} for the automotive industry has a bandwidth of up 
to 1 Mbps. Consequently, the copies only consume a small portion of 
bandwidth, making SWA cost-efficient. 
In addition, according to statistics \cite{murray2012state}, 
in Ethernet, packets using Type of Service and Differentiated 
Service Code Point account for 18\% of the enterprise network traffic. 
So, the copies set as high priority traffic will always have a high 
probability to improve the transmission of TT frames with BE 
delivery. Finally, to control the {\em jitter} (defined as the 
difference of the maximum and minimum latency) of a delivered frame,
we extend the preliminary architecture in \cite{Li2019DOE} to support 
a configured jitter by holding the frame until the jitter is less 
than a configured range.

We have implemented the proposed architecture in our industrial TT 
switches with Xilinx FPGA Virtex-7 XC7VX485T. Four test scenarios 
are presented to demonstrate the improvement and controllable jitter of 
the e2e latency of TT frames against the state-of-the-art TT 
transmission \cite{SteinerSMTScheduler, Li2020Time}. 
This paper makes the following main contributions:
\begin{itemize}
\item Proposal of a synergistic switch architecture (SWA) and 
its five forwarding steps as well as their respective algorithms 
covering the entire forwarding process of the copies of TT frames 
from ingress into a switch to egress from the switch.
\item Four properties of the SWA that dynamically improve the e2e latency 
of TT frames by opportunistically exploiting BE transmission.
\item Extension of the architecture to support controllable jitters by 
holding those delivered frames until the jitters fall in their 
configured ranges.
\item Implementation of the SWA in our FPGA-based industrial TT switches 
and demonstration of four test scenarios to validate its dynamical 
latency improvement and controllable jitter for TT frames compared 
to the current TT transmission.
\end{itemize}

The rest of the paper is organized as follows. Section II discusses 
related work and describes the SWA's novelty. Section III describes 
the background for TT transmission and the problems of designing such a SWA. 
Section IV details the design, algorithms and properties of the SWA. 
Section V evaluates the switch architecture for four test scenarios
and compares it with the current TT transmission. 
Finally, Section VI concludes the paper.

\section{Related Work}

Although time-triggered networks integrate TT and BE transmissions,
the two types of transmission are still independent in the switch 
architecture for transmission of their respective frames. 
802.1Qbv \cite{IEEE8021Qbv} employs a gate control list at the 
scheduled sending time instants to separate TT frames from 
BE frames, and uses guard-band or preemption strategies 
\cite{IEEE8021Qbu} to reduce the effect of BE frames on TT frames. 
\cite{Li2020Time, Yan2020TSNBuilder, Vlk2020Enhancing} follow 
the switching scheme of 802.1Qbv. \cite{Yan2020TSNBuilder} 
proposes a template-based TSN-builder to customize TSN switches 
according to an application-dependent resource abstract. \cite{Vlk2020Enhancing} co-designs FIFO scheduling constraints 
with hardware sequence checking of TT frames to improve usage 
rate of queues. Furthermore, \cite{Li2020Time} presents a 
memory-switch architecture and shared-memory scheduling 
constraints to make all ports share on-chip memory. 
To the best of our knowledge, the proposed SWA is 
the first cooperative TT and BE transmissions. It uses BE 
transmission to collaboratively deliver clone copies of TT frames 
to improve the e2e latency of TT transmission. It is different 
from the multipath redundancy of 802.1Qca \cite{IEEE8021Qca} which 
uses multiple redundant paths for TT frames and all 
paths are used to transmit TT frames with TT transmission. 
\cite{Su2019Synthesizing} presents a multipath redundant 
scheduler for TT transmission.

TT transmission depends on TT scheduling \cite{SteinerSMTScheduler, 
pozoSMTScheduler15p, pozoSMTScheduler15, craciunas2016combined, 
CraciunasTSNScheduling, Steiner2018Traffic, SernaOliver2018IEEE} 
to plan the precise sending instants of TT frames in each 
of end-devices and switches according to the 
application-specific e2e latency requirements, especially for 
industrial control. Scheduling results are transformed to 
schedule tables as the configuration of switches. Based on 
schedule tables, switches send TT frames at their precise sending 
instants for real-time control. \cite{SteinerSMTScheduler} first
formalizes TT scheduling by linearizing constraints such as 
e2e delay, free contention, and limited buffers, and then uses 
a satisfiability-modulo-theories (SMT) solver \cite{deSMT} to 
iteratively solve the scheduling problem.
\cite{pozoSMTScheduler15p} tunes the configuration parameters of 
SMT solvers for enhancing performance. \cite{pozoSMTScheduler15, 
craciunas2016combined} decomposes the scheduling constraints 
into multiple subsets and solves them incrementally to reduce 
the solution time. Especially, \cite{CraciunasTSNScheduling, 
Steiner2018Traffic, SernaOliver2018IEEE} linearize the constraints 
of 802.1Qbv. \cite{SernaOliver2018IEEE} uses the first-order 
theory of arrays to directly generate the gate control list of 
802.1Qbv. Recently, heuristic strategies \cite{NayakTSN, 
Yu2018Fast, Nayak2018Incremental, Wang2019Adaptive, Yu2020Adaptive, Falk2020Time, Jia2021TTDeep} have been proposed to speed up TT scheduling. 
In addition, BE frames are transmitted in the interval between 
TT frames. To improve the quality of service (QoS) of BE 
transmission, \cite{Tamas-Selicean2012, Tamas-Selicean2014} use 
a Tabu Search-based meta-heuristic algorithm to adjust the 
interval between TT frames. Such an interval adjustment even worsens 
the e2e latency of TT frames. All these schedulers generate the 
precise sending instants for TT frames in each device. They 
prevent the transmission of TT frames as soon as possible since 
they must wait until their time instants.

Our proposed SWA can in theory improve any 
scheduler for TT transmission, namely, both previous  
and future TT schedulers that modify the 
previous schedulers to minimize the e2e latency of TT frames, 
since it opportunistically exploits BE transmission to deliver 
TT frames as soon as possible. In other words, SWA 
can improve the e2e latency of TT frames as long as 
BE transmission is possible.

\section{Background and Problem Definition}
This section first presents the background for TT transmission, 
and then details the problems in designing a synergistic 
switch architecture.

\subsection{Background}
\subsubsection{Basic terminology and Concepts}

We model the topology of a network as an undirected graph 
\textit{G(V,E)}, where vertices \textit{V} represent the end-systems 
and switches and edges \textit{E} represent the physical communication 
links connecting vertices. Fig.~\ref{fig:TTNetworkExample} shows an 
example network topology with 8 vertices including 6 end-systems and 
2 switches. An ordered tuple $[v_{i}, v_{j}], v_{i}, v_{j} \in V$ 
defines a directed ``dataflow links" from $v_{i}$ to $v_{j}$. 
A sequence of dataflow links $l_{i}$ forms a ``dataflow path". 
An example of a dataflow path from A to F is depicted by the dotted 
line in Fig.~\ref{fig:TTNetworkExample}. We formally express a 
dataflow path $p$ from a sender $v_{0}$ to a receiver ${v_{n+1}}$ 
by the sequence of its dataflow links:
\begin{displaymath}
p = [[v_{0}, v_{1}],\ldots,[v_{n}, v_{n + 1}]],
\end{displaymath}
where the dataflow path has $n$ switches (i.e.,
$v_{1},v_{2},\ldots,v_{n}$). Thus, a dataflow path defines
a route from a sender to exactly one receiver. 

Information between the sender and the receiver is communicated in 
the form of TT flows that are composed of periodic TT frames 
according to AS6802. Let $F$ denote the set of TT flows. 
A flow $f_{i}\in F$ on a dataflow link $[v_{k},v_{l}], f_{i}^{[v_{k}, 
v_{l}]}$ is temporally specified by the following quadruple:
\begin{displaymath}
f_{i}^{[v_{k},v_{l}]} = \{f_{i}.period, f_{i}^{[v_{k}, v_{l}]}.offset, 
f_{i}.length, f_{i}.sequence\}.
\end{displaymath}
The period and length of a flow are specified by the underlying 
application. The flow sequence identifies different TT frames of 
the same flow in different periods. $f_{i}^{[v_{k},v_{l}]}.offset$ 
is the departure time of flow $f_{i}$ from vertex $v_{k}$ to vertex 
$v_{l}$ and is assigned by TT schedulers.

\subsubsection{Time-triggered transmission} 
Scheduling results are transformed into schedule tables stored 
in devices. These devices send flows at the specific times according 
to schedule tables. For example, the $n$-th departure time of 
flow $f_{i}$ from vertex $v_{k}$ to $v_{l}$ is specified by 
$n*f_{i}^{[v_{k},v_{l}]}.period + f_{i}^{[v_{k},v_{l}]}.offset$. 
However, besides TT frames, TT networks also transmit BE frames.
Two typical methods have been used to avoid the conflict with BE frames. 
One is non-preemptive with a guard band in front of each TT frame 
transmission (according to the IEEE 802.1 Qbv-2015, Amendment 25: 
Enhancements for Scheduled Traffic), which assures the BE transmission can be done 
before transmitting TT frames. The other is preemptive 
with minimal guard band (according to the IEEE 802.1 Qbu-2016, 
Amendment 26: Frame Preemption), which minimizes the waiting time of 
TT frames in case of conflict. So, in TT networks, we assume all 
flows are sent with no waiting at the time of their departure.

Actually, given no-wait transmission, the processing delay from 
the departure time $f_{i}^{[v_{k},v_{l}]}.offset$ to the first bit 
of flow $f_{i}$ on the dataflow link $[v_{k},v_{l}]$ is constant, 
denoted by $pdelay^{[v_{k},v_{l}]}$. We also let $ldelay^{[v_{k},v_{l}]}$
be the link delay on $[v_{k},v_{l}]$, which is measured 
dynamically by the peer delay mechanism of the IEEE 1588. 
Hence, the time of arrival at vertex $v_{l}$ for flow $f_{i}$ is:
\begin{displaymath}
f_{i}^{[v_{k}, v_{l}]}.arrival = f_{i}^{[v_{k}, v_{l}]}.offset + 
pdelay^{[v_{k},v_{l}]} + ldelay^{[v_{k},v_{l}]}.
\end{displaymath}
However, TT networks depend on time synchronization whose 
accuracy, denoted by $\mu$, is defined as the maximum time 
difference of any two synchronized devices in the network. 
In theory, the jitter of the arrival time of flow $f_{i}$ is 
in the range $[f_{i}^{[v_{k}, v_{l}]}.arrival - \mu, 
f_{i}^{[v_{k}, v_{l}]}.arrival + \mu]$. In practice, depending on 
the implementation, the jitter is affected by the processing 
jitter, queuing policies, etc. Hence, the time of flow $f_{i}$ 
arriving at switch $v_{l}$ is in a time window denoted by the 
range $[f_{i}^{[v_{k}, v_{l}]}.\textit{arrival-start}, 
f_{i}^{[v_{k}, v_{l}]}.\textit{arrival-end}]$. 
Assuming the first dataflow link and the last dataflow link of 
flow $f_{i}$ are $[v_{0}, v_{1}]$ and $[v_{n}, v_{n+1}]$, 
respectively, the e2e latency of flow $f_{i}$ ranges from 
$f_{i}^{[v_{n}, v_{n+1}]}.\textit{arrival-start} - f_{i}^{[v_{0}, 
v_{1}]}.offset$ to $f_{i}^{[v_{n}, v_{n+1}]}.\textit{arrival-end} 
- f_{i}^{[v_{0}, v_{1}]}.offset$. 
The latency jitter is $f_{i}^{[v_{n}, v_{n+1}]}.\textit{arrival-end} 
- f_{i}^{[v_{n}, v_{n+1}]}.\textit{arrival-start}$.

\subsection{Problems of Designing A SWA}
TT transmission tends to hold frames until their scheduled sending 
time even when they can be transmitted right away.
In contrast, BE transmission delivers frames 
as soon as possible but its uncertainties in rerouting, 
congestion, and queuing do not ensure the satisfaction of e2e 
frame latency requirement. So, designing a SWA 
must address the uncertainty of BE transmission to improve 
TT transmission.
First, TT transmission is order-preserving due to the scheduled 
TT frames. However, rerouting may lead to 
out-of-order frames in BE transmission while it is easy to 
preserve delivery order with static routing. SWA keeps copies of TT 
frames along with the same paths as those of TT frames with static 
routes. Frame loss can also lead to out-of-order delivery if
the frame copy of a TT flow at the $i$-th period is lost 
due to congestion but the copy at 
the $(i+1)$-th period over-takes the $i$-th TT frame. 
SWA employs a sequence-based order-preserving 
strategy to drop such out-of-order copies and restore 
the right sequence.

Second, TT transmission guarantees the bounded e2e latency 
by transmitting TT frames at their scheduled precise sending 
instants. However, in BE transmission, queuing will delay the 
forwarding of copies of TT frames. As a result, a longer latency 
may incur to TT frames. Furthermore, frame 
loss will result in unreachable copies. SWA ensures 
that the latency of TT frames is the worst-case e2e 
latency because either a TT frame or its copy, whichever arrives first 
at the destination device, will be kept while discarding the other. 
Moreover, SWA allows frame loss. For example, a switch has frame loss 
due to congestion but the remaining path from its next switch to 
the destination is available. The proposed SWA can recover the copy from 
the next switch to continue improving the remaining latency of TT frames.

Finally, TT transmission guarantees the latency jitter to be in 
the time window determined by the time synchronized accuracy, the 
processing jitter of a switch, etc. However, in BE transmission, 
when a frame is forwarded without congestion in all devices along 
its path, the frame has the minimum e2e latency. 
As a result, SWA reduces the lower bound of the e2e 
latency to the minimum. The jitter of a TT frame is the 
difference between the minimum latency of its copy and the latency 
of the TT frame with the worst-case e2e latency in SWA. 
To combat the negative effects of the jitter, we 
extend the architecture to support a configurable jitter. 
The extended architecture makes a tradeoff between the jitter and the 
latency according to the application requirements. 
That is, the smaller the jitter, the higher the lower bound 
of the latency.

\begin{figure*}[htbp]
  \centering
  \includegraphics[width=7.0in]{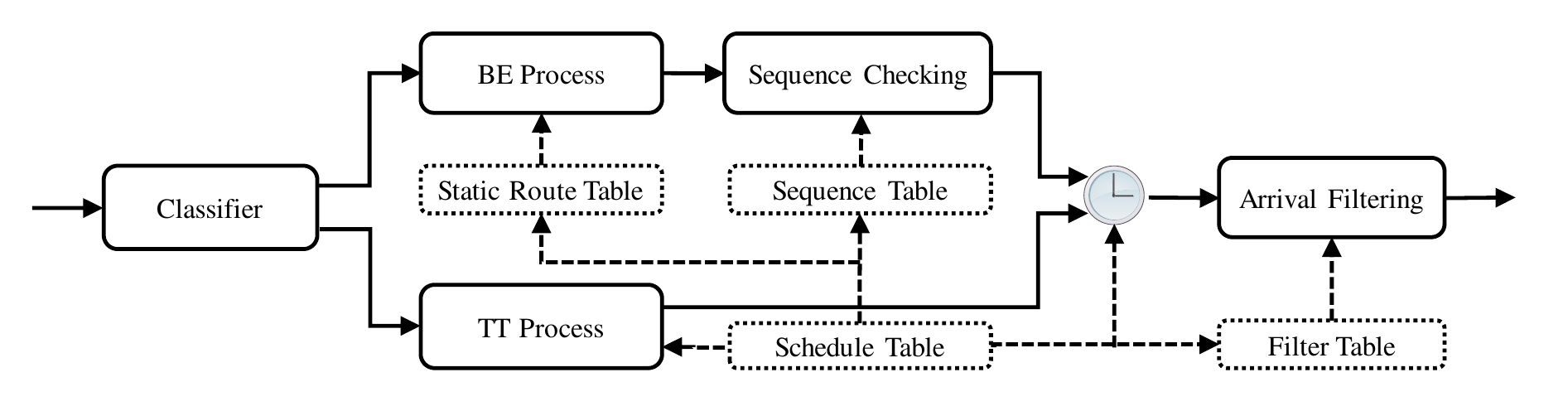}
  \caption{The synergistic switch architecture. Each solid 
  rectangle means one processing step for the architecture. 
  Each dotted rectangle means one table used by the corresponding 
  processing step. The solid rays indicate data flows (frames). 
  The dotted rays indicate control flows (table data).}
  \label{fig:best-tt-protocol-arch}
\end{figure*}

\section{Synergistic Switch Architecture}
The SWA uses BE transmission to enhance the performance of TT transmission. 
It defines the forwarding paradigm from ingress into a switch to egress 
from the switch for TT frame copies, as illustrated in 
Fig.~\ref{fig:best-tt-protocol-arch}. The solid rays indicate 
data flows (frames) and the dotted rays indicate control flows 
(table data). The forwarding paradigm contains four tables 
with dotted rectangles and five processing steps with solid 
rectangles, which is different from the standard forwarding 
scheme \cite{Li2020Time} that does not handle the uncertainty 
of BE transmission.

\begin{figure}[htbp]
  \centering
  \includegraphics[width=\linewidth]{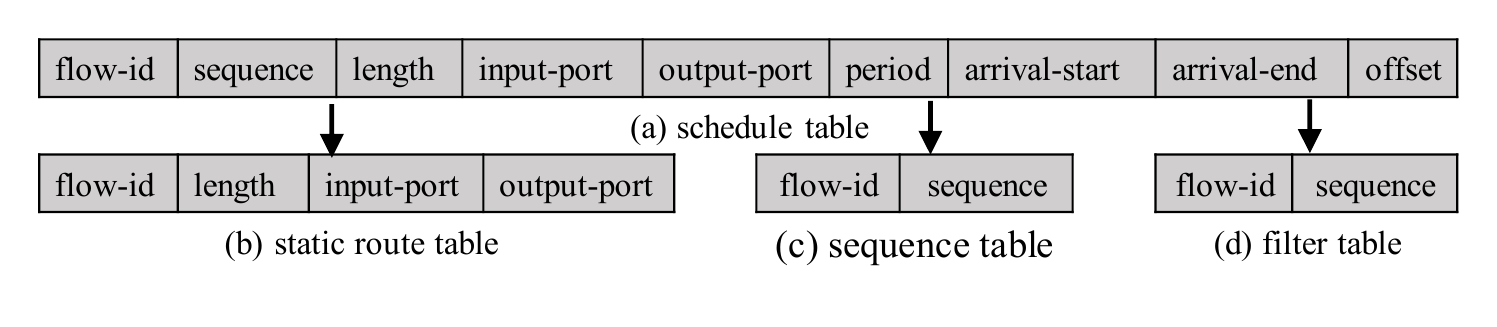}
  \caption{Tables in the SWA. 
  (a) is the content of the entire schedule table. (b), (c) and 
  (d) are \textit{static-route-table}, \textit{sequence-table} 
  and \textit{filter-table}, respectively, and their 
  initialization sources from \textit{schedule-table}.}
  \label{fig:tables}
\end{figure}

\subsection{Forwarding Tables}

TT flows are scheduled to have the precise departure times and 
the scheduled results are transformed into schedule tables 
in devices. The content of the schedule table 
\cite{Li2020Time} is illustrated as Fig.~\ref{fig:tables}(a). 
\textit{flow-id} and \textit{sequence} identify a unique flow and 
a unique TT frame of the flow, respectively. 
The \textit{sequence} also indicates the order of TT frames 
appearing in a flow. We assume that the sequence number 
is predictable. That is, given the current sequence number, the 
expected next sequence number can always be calculated. 
Furthermore, for simplicity, the expected next is equal to the 
current plus one. \textit{sequence} is updated 
based on the forwarding processes of the proposed architecture. 
\textit{period} is the time interval of a regularly repeating 
flow. For example, if the period of a flow is 1 $ms$, its frame 
will be transmitted every 1 $ms$. \textit{length} is the valid 
payload as bytes of a flow. The valid payload does not contain 
the cyclic redundancy check (4 bytes), preamble and start frame 
delimiter (8 bytes), and the shortest inter-frame gap (12 bytes). 
\textit{input-port} and \textit{output-port} are the input port 
and the output port, respectively, in a switch according to the 
dataflow path of a flow. \textit{arrival-start} and 
\textit{arrival-end} is the time window of the flow arriving at 
a switch. \textit{offset} is the departure time of a flow 
from a switch.

Copies of TT frames are transmitted using BE strategies. 
To ensure that copies follow the same dataflow paths as those of 
TT frames, they are forwarded according to a static route table which is 
illustrated in Fig.~\ref{fig:tables}(b). \textit{flow-id} is 
treated as the index to access the static route table. 
\textit{length} and \textit{input-port} are used to check a 
copy frame. Any discordance will lead to drop the copy frame. 
\textit{output-port} is the forwarding port of the copy. To 
preserve the order of copies in a flow, the sequence numbers 
of copies are checked with a sequence table illustrated in 
Fig.~\ref{fig:tables}(c). Finally, to make sure that only one 
frame, either a TT frame or its copy, is received by the destination
device, we drop the later frame and deliver the first arrival 
frame by a filter table as illustrated in 
Fig.~\ref{fig:tables}(d), in the output port of the last switch 
of a dataflow path. Fig.~\ref{fig:tables} presents the contents 
of all tables. (b), (c), and (d) are newly added for the proposed 
architecture and their initialization sources from (a). 
The updates of data fields in these tables are based on the 
forwarding processes of the proposed architecture.

\subsection{Forwarding Processes}

Fig.~\ref{fig:best-tt-protocol-arch} shows the overall 
proposed switch architecture. Following the dataflow in the 
architecture, five processing steps are performed to handle 
TT frames and their copies. The \textit{Frame} is the basic data 
structure for the architecture, which at least contains these 
data fields, namely \textit{flow-id}, \textit{sequence}, 
\textit{length}, \textit{input-port}, \textit{arrival-time} and 
\textit{iscopy}. \textit{arrival-time} is recorded by timestamps 
when the frame arrivals at a device. \textit{iscopy} equal to 
true indicates a cloned copy, else a TT frame.

\begin{algorithm}[t]
\KwIn{FIFO$<$Frame$>$ frames}
\KwOut{FIFO$<$Frame$>$ TTs, TTcopies}
     \While{!frames.empty()}{
         Frame frame = frames.dequeue();\\
         \If{frame.iscopy} {
             TTcopies.enqueue(frame);
         } \Else {
             Frame newframe = copy(frame);\\
             newframe.iscopy = true;\\
             TTs.enqueue(frame);\\
             TTcopies.enqueue(newframe);
         }
     }
\caption{Classifier}
\label{alg:classifier}
\end{algorithm}

\begin{algorithm}[t]
\KwIn{FIFO$<$Frame$>$ TTcopies}
\KwOut{FIFO$<$Frame$>$[] TTcopiesbyport}
\While{!TTcopies.empty()}{
     Frame frame = TTcopies.dequeue(TTcopies); \\
     Row row = static-route-table[frame.flow-id]; \\
     \If{(row.length == frame.length) and (row.input-port == frame.input-port)}{
        TTcopiesbyport[row.output-port].enqueue(frame);
     } \Else {
        drop frame;
     }
   }
\caption{BE process}
\label{alg:beprocess}
\end{algorithm}

\textbf{Step 1: Classifier} is illustrated in 
Alg.~\ref{alg:classifier}. It classifies the received frames by 
the data field \textit{iscopy}. If a frame is a copy, then
proceed to Step 2 \textit{BE process}. If a frame 
is a TT frame, the frame is first cloned to generate a new copy 
and then move to Step 3, \textit{TT process}. Its copy procceds
to Step 2 \textit{BE process}. The cloning happens 
at each switch.

\textbf{Step 2: BE process} is illustrated in 
Alg.~\ref{alg:beprocess}. It searches \textit{static-route-table} 
by \textit{flow-id} and then checks \textit{length} and 
\textit{input-port}. If matched, the frame is forwarded to the 
corresponding output port based on \textit{output-port}. 
Otherwise, the frame is dropped.

\begin{algorithm}[t]
\SetKwFunction{SetTimer}{SetTimer}
\KwIn{FIFO$<$Frame$>$ TTs}
\KwOut{FIFO$<$Frame$>$[] framesbyport}
\While{!TTs.empty()}{
     Frame frame = TTs.dequeue(); \\
     Row row = schedule-table[frame.flow-id];\\
     \If{(frame.arrival-time $\leq$ row.arrival-end) and (frame.arrival-time $\geq$ row.arrival-start) and (row.length == frame.length) and (row.input-port == frame.input-port) and (frame.sequence $>$ row.sequence)}{
         \SetTimer{row.offset} \{ \\
         row.sequence = frame.sequence;\\
         framesbyport[row.output-port].enqueue(frame);\\
         Row strow = sequence-table[frame.flow-id];\\
         \If{frame.sequence $>$ strow.sequence}{
             strow.sequence = frame.sequence;
          }
          \};
     } \Else {
         drop frame;
     }
   }
\caption{TT process}
\label{alg:ttprocess}
\end{algorithm}

\textbf{Step 3: TT process} is illustrated in 
Alg.~\ref{alg:ttprocess}. It searches \textit{schedule-table} by 
\textit{flow-id}. The arrival time, \textit{arrival-time}, of 
the frame must range from \textit{arrival-start} to 
\textit{arrival-end}. The sequence number, \textit{sequence}, of 
the frame must be larger than that in \textit{schedule-table}, 
indicating a new TT frame. \textit{length} and 
\textit{input-port} are also checked. If any unmatched, the frame 
is dropped. If all matched, a timer is set to $\textit{offset} 
+ m * period$ for the $m$-th TT frame. At that time instant, the 
sequence number of \textit{schedule-table} is updated and the 
frame is forwarded to the corresponding output. The sequence 
numbers of TT frames are used to restore the sequence 
numbers in \textit{sequence-table} for the cloned copies due to the 
uncertainties of the BE transmission, such as congestion and queuing. 
When the sequence number of the TT frame is larger than that in 
\textit{sequence-table}, it indicates that the cloned copies 
fall behind the TT frame due to congestion and queuing, and thus 
the sequence number in \textit{sequence-table} is updated so as
to drop those copies falling behind in Step 4, 
\textit{sequence checking}. Although multiple timers may be set, 
these operations are still conflict-free due to the TT scheduling 
that ensures different TT frames are not forwarded to 
the same port at the same time.

\textbf{Step 4: Sequence checking} is illustrated in 
Alg.~\ref{alg:sequencechecking}. It searches 
\textit{sequence-table} by \textit{flow-id} and checks 
if the current sequence number of the frame 
\textit{frame.sequence} is equal to the expected sequence number 
$row.sequence + 1$. If equal, the sequence number in 
\textit{sequence-table} is updated and the frame is passed to the 
corresponding output port, else the frame is dropped. 
Such a checking for sequence numbers ensures that cloned 
copies strictly follow the sequence number one-by-one, and thus 
is order-preserving. Any uncertainty such as surpassing or falling 
behind will lead to dropping of copies in the sequence checking. 
If such an uncertainty happens, the sequence numbers in 
\textit{sequence-table} are restored by Step 3, 
\textit{TT process}. TT copies sent as BE frames do not create any
conflict with TT frames due to the use of a guard band for TT transmission, 
and thus the updates of sequence numbers by this step and Step 3, 
\textit{TT process}, do not have any conflict.

\begin{algorithm}[t]
\KwIn{FIFO$<$Frame$>$[] TTcopiesbyport}
\KwOut{FIFO$<$Frame$>$[] framesbyport}
\For{i = 0; i $<$ TTcopiesbyport.num; i++}{ 
     \While{!TTcopiesbyport[i].empty()}{
         Frame frame = TTcopiesbyport[i].dequeue(); \\
         Row row = sequence-table[frame.flow-id]; \\
         \If{frame.sequence $==$ row.sequence + 1}{
              row.sequence = frame.sequence; \\
              framesbyport[i].enqueue(frame);
         }\Else{
              drop frame;
         }
     }
   }
\caption{sequence checking}
\label{alg:sequencechecking}
\end{algorithm}

\begin{algorithm}[t]
\KwIn{FIFO$<$Frame$>$[] framesbyport}
\KwOut{FIFO$<$Frame$>$[] deliveredframes}
\For{i = 0; i $\leq$ framesbyport.num; i++}{
     \While{!framesbyport[i].empty()}{
         Frame frame = framesbyport[i].dequeue();\\
         Row row = filter-table[frame.flow-id];\\
         \If{frame.sequence $>$ row.sequence}{
             row.sequence = frame.sequence; \\
             deliveredframes[i].enqueue(frame);
          } \Else {
             drop frame;
          }
     }
   }
\caption{Arrival Filtering}
\label{alg:arrivalfiltering}
\end{algorithm}

\textbf{Step 5: Arrival filtering} is illustrated in 
Alg.~\ref{alg:arrivalfiltering}. It searches \textit{filter-table} 
by \textit{flow-id} and checks the sequence number of the frame. 
If the sequence number is larger than that in 
\textit{filter-table}, the frame is delivered to the 
corresponding output port and updates the sequence number in 
\textit{filter-table}. Otherwise, the frame is dropped. 
This step performs filtering the later-arriving frames to ensure 
only the earlier piece, namely either a TT frame or its copy, 
is delivered. It is usually enabled at the latest switch that 
is directly connected to the end-device in a route. Depending on
demands, the data field \textit{iscopy} may be recovered to false 
so that the end-device treats it as a TT frame. If this step is 
disabled, the end-device may receive both a TT frame and its copy 
without the arrival filtering.

In addition, the processing steps, namely, Steps 3--5,
depend on the incremental sequence numbers. So, it is 
necessary to ensure the initial value of a sequence number in 
tables is equal to the minimal value like 0. When devices are 
initialized, or configurations are updated, the sequence number 
should be reset.

\subsection{Guaranteed Properties}

To improve the e2e latency of TT transmission via BE transmission, 
the underlying architecture should have the following four 
properties:
\begin{itemize}
    \item \textit{Preserving Order}: TT transmission is typically 
    for industrial control data whose sequence is rigidly 
    constrained by industrial requirements. The order of TT 
    frames is met by scheduling their transmission at precise 
    instants. So, the architecture should not alter the order of transmitting TT frames in spite of the uncertainty of 
    BE transmission.
    \item \textit{Self-Recovery}: The queuing and frame loss of 
    BE transmission due to congestion are usually short-lived 
    and partial. So, the architecture should recover the 
    transmission when the congestion disappears.
    \item \textit{Improvement}: The architecture should improve 
    the e2e latency of TT frames.
    \item \textit{Cost-Efficiency}: The overhead of the architecture should be low.
\end{itemize}

The proposed SWA provides four salient properties as follows.

\begin{property}[\textbf{Preserving Order of Transmitting TT Frames}]
Different TT frames in a flow from the origin device are 
received by the end-devices in the same order the origin
device sent.
\end{property}
\begin{proof}
Assuming the architecture is not order-preserving,
the order of transmitting TT frames is not the same
as the order of their reception. 
Let $f_{i}$ denote the TT frame with the sequence number $i$. 
All transmitted frames have a strictly increasing sequence 
numbers, i.e., $f_{i}$ is sent earlier than $f_{j}$ 
$\forall~i < j$. 
Due to the difference between the sent and the received order, 
whatever surpassing or falling behind, there exist two 
contiguously received frames, $f_{p}$ and $f_{q}$, $p > q$, but 
$f_{p}$is are received before $f_{q}$. Moreover, since $f_{p}$ 
and $f_{q}$ belong to the same flow with different sequence 
numbers, they travel on the same static route and data flow path.

If both $f_{p}$ and $f_{q}$ are copies, upon receipt of $f_{p}$ 
by an end-device, the sequence numbers of \textit{sequence-tables} 
on the data flow path are all updated to $p$. 
Checking $frame.sequence == row.sequence + 1$ will lead to 
dropping of $f_{q}$.

If $f_{p}$ and $f_{q}$ are both TT frames, after $f_{p}$ is 
received by an end-device, the sequence numbers of 
\textit{schedule-tables} along with the data flow path are all 
updated to $p$. Çhecking with $frame.sequence > row.sequence $ 
will lead to dropping of $f_{q}$.

If either $f_{p}$ or $f_{q}$ is a copy and the other is 
a TT frame, after an end-device receives $f_{p}$, 
the sequence number of \textit{filter table} in the data flow 
path is updated to $p$. Checking with $frame.sequence > 
row.sequence $ will lead to dropping of $f_{q}$.

In all of the above cases, $f_{q}$ will be dropped,
a contradiction since $f_{q}$ is received. So, the 
synergistic architecture is order-preserving.
\end{proof}

Although bandwidths are reserved at specific sending instants 
for TT frames via their scheduling and TT frames do not have conflict 
with BE frames under the guard-band strategy, they may still be 
dropped since the network environment is dynamic and complex 
such as link or switch failures. 
So, the condition $frame.sequence > row.sequence$ is checked
in \textit{TT process} and \textit{arrival filtering} to recover 
the transmission of TT frames whenever a new TT frame arrives. 
However, the \textit{sequence checking} step checks sequence numbers 
with $frame.sequence == row.sequence + 1$ because it is used for the 
transmission of copies of TT frames with the uncertainty of BE 
transmission. For example, considering the case when a copy 
with the sequence number $m$ is dropped at a switch $v$ due to 
congestion while the next copy with sequence number $m + 1$ 
overtakes the TT frame with sequence number $m$ and also arrives 
at $v$, if \textit{sequence checking} 
also uses $frame.sequence > row.sequence$, the copy with sequence number
$m + 1$ will be forwarded. As a result, the sequence number in 
\textit{filter-table} is updated to $m + 1$, and thus the TT frame 
with sequence number $m$ will be filtered out. So, checking 
$frame.sequence > row.sequence$ in \textit{sequence checking} 
will lead to dropping of TT frames, which is unexpected because 
BE transmissions are introduced to make TT transmissions as soon 
as possible, and thus should not hurt the transmission of 
TT frames.

Hence, we use the checking of $frame.sequence == row.sequence + 1$ 
in \textit{sequence checking}. But, it may drop all the following 
copies if a copy is dropped. To deal with such a situation and 
recover the transmission of the copies, we first 
establish the following lemma.
\begin{lemma}
\label{l1}
For a flow $f \in F$, $p=[[v_{0},v_{1}],\dots,[v_{n},v_{n+1}]]$ 
is a dataflow path of $f$, the constant $C$ is the minimum 
forwarding time in a single switch. 
If $f^{[v_{n},v_{n+1}]}.offset - f^{[v_{0},v_{1}]}.offset < 
f.period + n * C$, we have $f^{[v_{k},v_{k+1}]}.offset - 
f^{[v_{0},v_{1}]}.offset < f.period + k * C$ for all 
dataflow links, $[v_{k}, v_{k+1}] \in p, 0 \leq k \leq n$.
\end{lemma}
\begin{proof} Since $C$ is the minimum forwarding time in 
a single switch. We have \\ 
$ (n-k) * C \leq f^{[v_{n},v_{n+1}]}.offset - f^{[v_{k},v_{k+1}]}.offset$. \\
By $ f^{[v_{n},v_{n+1}]}.offset - f^{[v_{0},v_{1}]}.offset 
= (f^{[v_{n},v_{n+1}]}.offset - f^{[v_{k},v_{k+1}]}.offset) + 
(f^{[v_{k},v_{k+1}]}.offset- f^{[v_{0},v_{1}]}.offset)$, we have \\
$ (n-k) * C + (f^{[v_{k},v_{k+1}]}.offset- f^{[v_{0},v_{1}]}.offset) 
\leq f^{[v_{n},v_{n+1}]}.offset - f^{[v_{0},v_{1}]}.offset$. \\
Since $ f^{[v_{n},v_{n+1}]}.offset - f^{[v_{0},v_{1}]}.offset < f.period 
+ n * C$, we have \\
$(n-k) * C + (f^{[v_{k},v_{k+1}]}.offset- f^{[v_{0},v_{1}]}.offset) 
< f.period + n * C$,  and thus \\
$f^{[v_{k},v_{k+1}]}.offset - f^{[v_{0},v_{1}]}.offset < f.period + k * C$.
\end{proof}

the minimum forwarding time in a single switch, $C$, can be achieved 
by testing, or by its product specification. $\forall f \in F$, 
we constrain the e2e latency as:
\begin{equation}
\label{e1}
\begin{split}
&f^{[v_{n},v_{n+1}]}.offset - f^{[v_{0},v_{1}]}.offset < \\
&min\{period + n * C, latency\}.
\end{split}
\end{equation}
Finally, we use previous algorithms to schedule these flows. 
The dropped copies in those scheduled flows will not affect 
the transmission of the following copies. 
The formal self-recovery property is given as:
\begin{property}[\textbf{Self-Recovery}]
If a flow $ f \in F$ satisfies Eq.~(\ref{e1}), then 
it is self-recoverable, i.e., if its copy with sequence
number $m$ is dropped, then the transmission of the next 
copy with sequence number $m + 1$ will not be affected.  
\end{property}
\begin{proof}
Suppose the transmission of the next copy is affected by 
the dropping of the copy with sequence number $m$ and let $v_{i}$ 
denote the switch where the copy is dropped. Since the next 
copy is affected, it is dropped as a result of checking 
$frame.sequence == row.sequence + 1$ in \textit{sequence 
checking}. That is, the TT frame with $m$ arrives at $v_{i}$ 
not earlier than the next copy with $m + 1$, else the TT 
frame with $m$ will update the sequence number in 
\textit{sequence-table} to $m$, and thus the next copy will 
not be dropped. So, the latency of the TT frame with $m$ 
from $v_{0}$ to $v_{i}$ is not smaller than the sum of 
\textit{period} and the forwarding time of the next copy 
with $m + 1$ from $v_{0}$ to $v_{i}$. Thus, 
$f^{[v_{i},v_{i+1}]}.offset - f^{[v_{0},v_{1}]}.offset 
\geq f.period + i * C$. But, the flow $f$ satisfies 
\autoref{e1}. So, by Lemma \ref{l1}, 
$f^{[v_{i},v_{i+1}]}.offset - f^{[v_{0},v_{1}]}.offset < 
f.period + i * C$, a contradiction. Thus, the transmission 
of the next copy with $m + 1$ will not be affected.
\end{proof}

\begin{property}[\textbf{Dynamic Improvement}]
The synergistic architecture makes a dynamic improvement 
of the e2e latency of TT frames.
\end{property}
The synergistic architecture makes dynamic improvements 
in the following three aspects.
\begin{enumerate}
\item The copy of each TT frame is independently transmitted 
as soon as possible according to the current network load. 
So, different copies, even in the same flow, usually have 
different e2e latencies.

\item If a copy is dropped at switch $v_{i}$, the latency of 
the remaining path from $v_{i+1}$ to the end-device will be 
improved by a new copy starting at switch $v_{i+1}$ because 
SWA attempts to create a new copy at each 
switch on the flow's path.

\item The latency of TT frames is an upper bound of the e2e 
latency because the \textit{arrival filtering} selects the 
first arrived frame, namely, either a TT frame or its copy 
as the final delivery and ignores the others.
\end{enumerate}

\begin{property}[\textbf{Cost-Efficiency}]
For each TT frame, no more than one copy of the TT frame is 
transmitted simultaneously in the whole network.
\end{property}
\begin{proof}
Since the strict checking of $frame.sequence == row.sequence 
+ 1$ in \textit{sequence checking}, after a copy of a TT frame 
is transmitted, $row.sequence$ will be updated to 
$row.sequence = frame.sequence$ by 
Alg.~\ref{alg:sequencechecking}. As a result, the other copies 
of the same TT frame will all be dropped as a result of checking 
$frame.sequence == row.sequence + 1$. Since \textit{sequence 
checking} happens at each switch, although each switch tries to 
copy a TT frame, other copies will be dropped by the switches 
that generated them after one copy is transmitted. 
So, in the whole network, no more than one copy of the TT 
frame is transmitted simultaneously.
\end{proof}

The above four properties enable the synergistic architecture 
to opportunistically exploit BE transmissions to dynamically 
improve the e2e latency of TT frames because they handle 
the uncertainty of BE transmissions and make the architecture
cost-efficient.

\subsection{Extension for Controllable Jitter}
Although the SWA optimizes the e2e latency 
of TT frames, it increases the latency jitter of TT 
transmissions. In general, given a flow $f$ along with a 
dataflow path $p=[[v_{0},v_{1}],\ldots,[v_{n},v_{n+1}]]$, 
the maximum jitter based on SWA is 
$f^{[v_{n},v_{n+1}]}.\textit{arrival-end} - n * C$ in theory 
where $C$ is the minimum forwarding time in a single switch. 
However, the jitter of TT transmission is 
$f^{[v_{n},v_{n+1}]}.\textit{arrival-end} - 
f^{[v_{n},v_{n+1}]}.\textit{arrival-start}$. To limit the 
jitter, we add a new data field $jitter$ into the filter table 
\textit{filter table} and follow a simple holding strategy to 
control the jitter. That is, when the flow $f$ arrives before 
$f^{[v_{n}, v_{n+1}]}.offset - jitter$, the switch $v_{n}$ holds 
it until $f^{[v_{n}, v_{n+1}]}.offset - jitter$; 
else the switch delivers the flow immediately. Such a strategy
of holding frames controls jitter according to the 
application requirement by configuring the data field 
$jitter$. In general, the configurable range of the jitter of  
flow $f$ is $[0,f.period]$ because the jitter larger than 
the period will lead to no frame or multiple frames in a period, 
which is unacceptable for industrial real-time control applications
\cite{Carvajal2014}. So, to tolerate such unacceptable 
configurations, i.e., $jitter \notin [0,f.period]$, 
we assume that $jitter < 0$ is to minimize the jitter, 
i.e., dropping all copies and delivering TT frames only, 
and $jitter > f.period$ represents no constraint on jitter.

However, holding frames until the jitter falls in a configured 
range fails their original timely transmission. 
As a result, the held frames may conflict TT frames, 
copies, and other BE frames. Such conflicts should be properly 
handled if the extension for controlled jitter is implemented. 
First, the held frames are copies of TT frames rather than TT 
frames themselves since the e2e latencies of TT frames are the 
worst-case values, and thus TT frames need not be held. 
Second, since the held copies are delivered at the specified 
sending instants according to the configured jitter, which is 
akin to TT transmission, the conflict resolutions used by TT 
transmission can also be applied to handle the conflicts 
between the held copies and other BE frames. 
For simplicity, it is even not necessary to handle the conflicts 
between the held copies and other BE frames because if other BE 
frames are casually transmitted at the specified sending 
instants of the held copies, the transmission of copies simply waits 
sent until the transmission of BE frames is completed. In the worst 
case, a cloned copy of a TT frame will wait until the transmission
time instant of the TT frame. As a result, such a wait yield
a smaller jitter. Finally, we provide a resolution 
to the conflicts between the held copies and other copies, 
TT frames by constraining the configurable range of a jitter. 
We define a safe-jitter range such that if the jitter 
is configured in the range, there will be no conflict between 
the held copies and other copies, TT frames.
Theorem \ref{t1} provides the safe-jitter range of each TT flow.

\begin{theorem}[\textbf{Safe Jitter Range}]
Let $F$ be the set of all flows having the same output port
and $N$ be the number of flows in F. $\forall f_{i} \in F, 0 \leq 
i < N$, let $jitter_{i}$ denote the jitter configuration of 
$f_{i}$. $g_{ij}$ is the minimum gap of the departure time 
of $f_{j}$ before the departure time of $f_{i}$. Especially, 
when $j=i$, $g_{ii}$ is $f_{i}.period$, 
indicating that the minimum gap of two TT frames of $f_{i}$ 
is the period of $f_{i}$. $C_{i}$ is the transmission time of 
$f_{i}$ under a given bandwidth.\footnote{For example, given 
the bandwidth, 100 Mbps, the transmission time per bit is 
10 ns.} So, the safe range of $jitter_{i}$ is:
\begin{equation}
0 \leq jitter_{i} \leq min\{g_{ij} - C_{j} \ |\ 0 \leq j 
\leq N - 1\}. 
\end{equation}
\label{t1}
\end{theorem}
\begin{proof}
First, TT frames are conflict-free as a result of scheduling
their transmission. So, given $\forall f_{i}, f_{j} \in F, 0 \leq 
i, j < N$, the transmission of $f_{j}$ can always finish 
before the departure time of $f_{i}$. That is, $g_{ij} - C_{j} 
\geq 0$. Therefore, $min\{g_{ij} - C_{j} \ |\ 0 \leq j 
\leq N - 1 \} \geq 0$.

Second, assuming $g_{ij} > f_{i}.period$, we set $g^{'}_{ij} = 
g_{ij} - f_{i}.period$. Thus, $g^{'}_{ij}$ is another gap of 
the departure time of $f_{j}$ prior to the departure time of 
$f_{i}$ and $g^{'}_{ij} < g_{ij}$, which is a contradiction 
since $g_{ij}$ is the minimum gap. So, $g_{ij} \leq 
f_{i}.period$. Especially, if $g_{ij}==f_{i}.period$ is 
established, then $f_{i}$ and $f_{j}$ are the same flow. 
Otherwise, $f_{i}$ and $f_{j}$ are different flows. 
Since $g_{ij}==f_{i}.period$, another frame of $f_{i}$ is 
transmitted at the gap $g_{ij}$ prior to the departure time 
of $f_{i}$. At the same time, a frame of $f_{j}$ is also 
transmitted at the gap $g_{ij}$ prior to the departure time 
of $f_{i}$ according to the definition of $g_{ij}$. 
So, a conflict between $f_{i}$ and $f_{j}$ occurs, which is a 
contradiction since TT frames are conflict-free. Furthermore, 
since $C_{j} > 0$, we have that $min\{g_{ij} - C_{j} \ |\ 0 
\leq j \leq N - 1 \} < f_{i}.period$.

Third, given $\forall f_{i}, f_{j} \in F, 0 \leq i, j < N$, 
since $0 \leq jitter_{i} \leq min\{g_{ij} - C_{j} \ |\ 0 \leq j 
\leq N - 1 \}$, we have that $0 \leq jitter_{i} 
\leq g_{ij} - C_{j}$. A held copy of $f_{i}$ will be 
transmitted at the gap $jitter_{i}$ prior to the departure 
time of flow $f_{i}$. So, the held copy will be transmitted 
after or at the gap $g_{ij} - C_{j}$ prior to the departure 
time of flow $f_{i}$. Thus, its transmission time does not 
overlap with that of TT frames of $f_{j}$. Hence, the held copy 
of $f_{i}$ has no conflict with TT frames of $f_{j}$. 
Since $f_{i}$ and $f_{j}$ are generic, the held copies have 
no conflicts with TT frames under $0 \leq jitter_{i} \leq 
min\{g_{ij} - C_{j} \ |\ 0 \leq j \leq N - 1 \}$.

\begin{algorithm}[t]
\KwIn{Flow fs[N], int C[N]}
\KwOut{int uppers[N]}
\For{i = 0; i $<$ N; i++}{
    uppers[i] = fs[i].period; \\
    \For{j = 0; j $<$ N; j++}{
        gap = fs[i].period; \\
        \For{k = 0; k $<$ $\frac{LCM(fs)}{fs[j].period}$; k++}{
            q = (k * fs[j].period + fs[j].offset) - fs[i].offset; \\
            \If{q $\geq$ 0}{
                m = $\lfloor \frac{q}{fs[i].period} \rfloor$;\\
                \If{(m + 1) * fs[i].period - q $<$ gap}{
                   gap = (m + 1) * fs[i].period - q;
                }
            } \Else {
                k = $\lfloor \frac{(fs[i].offset - fs[j].offset)}{fs[j].period} \rfloor$; \\
                \If{fs[i].offset - (k * fs[j].period + fs[j].offset) $<$ gap}{
                   gap = fs[i].offset - (k * fs[j].period + fs[j].offset);
                }
            }
        }
        \If{uppers[i] $>$ gap - C[j]}{
            uppers[i] = gap - C[j];
        }
    }
 }
\caption{Computation of safe Jitter}
\label{alg:jitter}
\end{algorithm}

Finally, given $\forall f_{i}, f_{j} \in F, 0 \leq i, j < N$, 
$tc_{ik}$ is a held copy of the TT frame of the $k$-th period 
of $f_{i}$ and $tc_{jl}$ is a held copy of the TT frame of the 
$l$-th period of $f_{j}$. Since the configurations, $jitter_{i}$ 
and $jitter_{j}$, the departure time of the held copy, $tc_{ik}$ 
is $k* f_{i}.period + f^{[v_{n}, v_{n+1}]}_{i}.offset - 
jitter_{i}$ while the departure time of the copy, $tc_{jl}$ is 
in the range from $l* f_{j}.period + f^{[v_{n}, 
v_{n+1}]}_{j}.offset - jitter_{j}$ to $l* f_{j}.period + 
f^{[v_{n}, v_{n+1}]}_{j}.offset$. Assuming $tc_{ik}$ and 
$tc_{jl}$ have conflicts, their transmission times overlap. 
So, the departure time of $tc_{ik}$ falls in the range from 
$l* f_{j}.period + f^{[v_{n}, v_{n+1}]}_{j}.offset - jitter_{j}$ 
to $l* f_{j}.period + f^{[v_{n}, v_{n+1}]}_{j}.offset + C_{j}$. 
That is, \\ $k* f_{i}.period + f^{[v_{n}, v_{n+1}]}_{i}.offset 
- jitter_{i} < l* f_{j}.period + f^{[v_{n}, 
v_{n+1}]}_{j}.offset + C_{j}$. 
Since $0 \leq jitter_{i} \leq min\{g_{ij} - C_{j} \ |\ 0 
\leq j \leq N - 1\}$, we have $jitter_{i} \leq g_{ij} - C_{j}$. 
Thus, 
$(k* f_{i}.period + f^{[v_{n}, v_{n+1}]}_{i}.offset) - 
(l* f_{j}.period + f^{[v_{n}, v_{n+1}]}_{j}.offset) < g_{ij}$. 
Setting $g^{'}_{ij}=(k* f_{i}.period + f^{[v_{n}, 
v_{n+1}]}_{i}.offset) - (l* f_{j}.period + f^{[v_{n}, 
v_{n+1}]}_{j}.offset)$, so $g^{'}_{ij}$ is another gap of the 
departure time of flow $f_{j}$ prior to the departure time of 
flow $f_{i}$ and $g^{'}_{ij} < g_{ij}$. According to the 
definition of $g_{ij}$, $g_{ij}$ is the minimum gap, which  
contradicts $g^{'}_{ij} < g_{ij}$. Hence, the held copy 
$tc_{ik}$ has no conflicts with the copy $tc_{jl}$. Since the 
generality of $tc_{ik}$ and $tc_{jl}$, the held copies will 
not conflict with other copies under $0 \leq jitter_{i} 
\leq min\{g_{ij} - C_{j} \ |\ 0 \leq j \leq N - 1\}$.

Given the jitter configuration $jitter_{i}$ of each flow 
$f_{i}$ satisfies $0 \leq jitter_{i} \leq min\{g_{ij} - C_{j} 
\ |\ 0 \leq j \leq N - 1\}$, the held copies have no conflicts 
with TT frames and other copies. 
\end{proof}

According to the proof of Theorem \ref{t1}, $jitter_{i}$ that 
satisfies $0 \leq jitter_{i} \leq min\{g_{ij} - C_{j} \ |\ 0 
\leq j \leq N - 1\}$ is also in $[0, f_{i}.period]$. 
So, Theorem \ref{t1} constrains the configuration range of 
jitter of each flow and ensures that the held copies have no 
conflicts with TT frames and other copies. Algorithm 
\ref{alg:jitter} treats the scheduled flows $fs$ and their 
transmission time $C$ as input and outputs the upper bound of 
the safe-jitter range of each flow, denoted by $uppers$, 
according to Theorem \ref{t1}. $N$ is the number of flows. 
$LCM(fs)$ denotes the least common multiple of these 
scheduled flows $fs$.

\begin{figure}[htbp]
  \centering
  \includegraphics[width=3.6in]{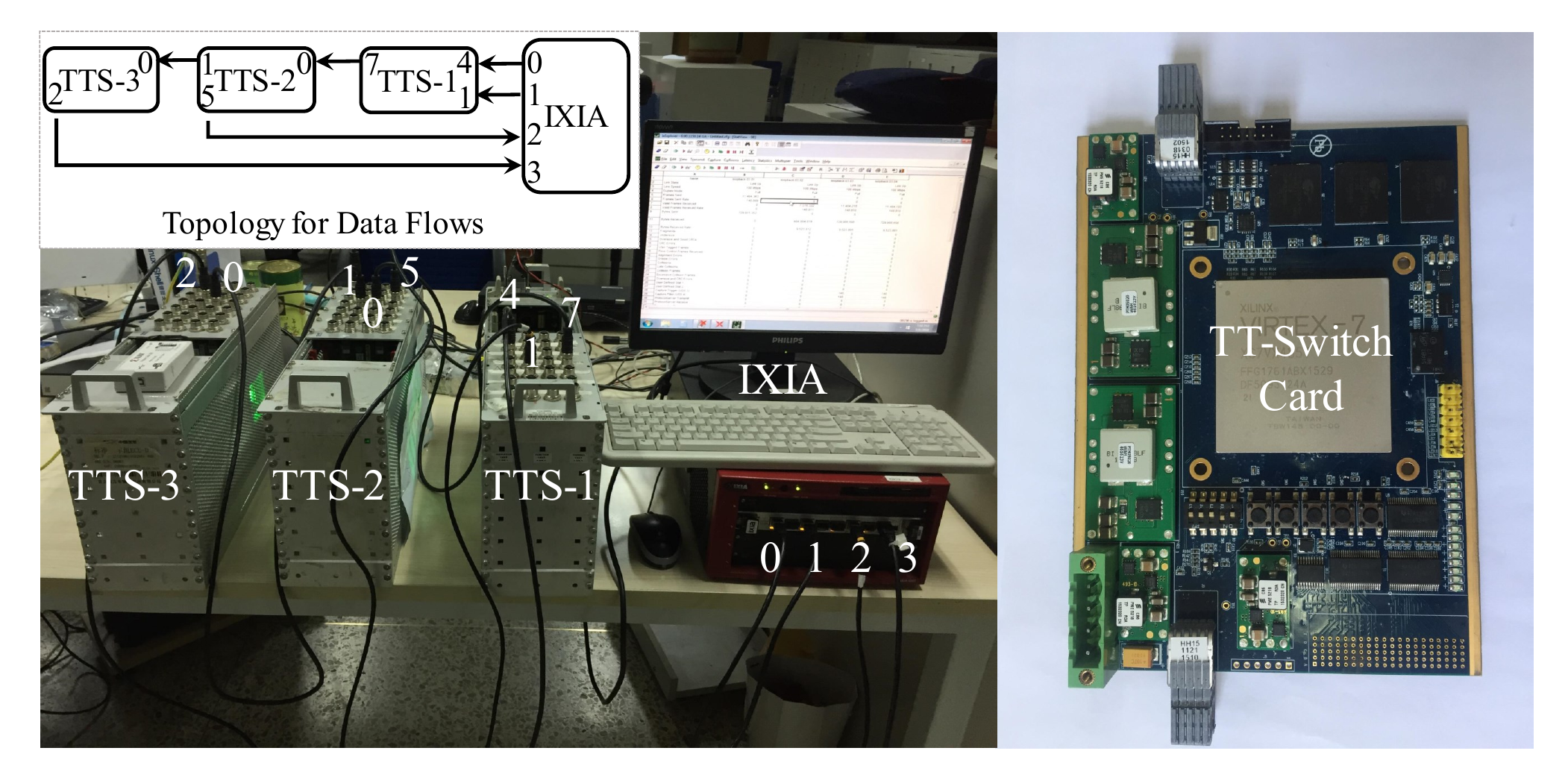}
  \caption{The experiment platform for the synergetic switch architecture. The \textit{IXIA} is a standard tester. The \textit{TTS-1}, \textit{TTS-2} and \textit{TTS-3} are our time-triggered (TT) switches. The \textit{TT-Switch Card} is our switch card installed in \textit{TTS-1}, \textit{TTS-2} and \textit{TTS-3}. The digitals are ports. The rays are links and constitute the topology for data flows.}
  \label{fig:platform}
\end{figure}

\section{Evaluation Results and Analysis}
To demonstrate the advantages of SWA, we compare it with the 
state-of-the-art TT transmission that uses the typical scheduler 
\cite{SteinerSMTScheduler} to schedule TT frames and uses 
FPGA-based TT switches \cite{Li2020Time} to transmit these 
frames. Based on the TT switches, we implement the proposed 
SWA in our TT switches illustrated 
in Fig.~\ref{fig:platform} with the Xilinx FPGA Virtex-7 
XC7VX485T to demonstrate the practicability 
of SWA. Our TT switches have $24$ fast-Ethernet 
(100 Mbps) ports. Fig.~\ref{fig:platform} shows the experiment 
platform including three TT switches and a standard tester. 
We designed four scenarios for comparison as follows.
\begin{itemize}
    \item \textbf{Scenario One (Latency Improvement)}: 
    We show that the copies of TT frames can always improve the e2e latency 
    of TT frames in spite of the disturbance of other BE data. So, we assign 
    copies a higher priority than that of the other BE data and compare the 
    e2e latency of SWA with that of the previous TT transmission.
    
    \item \textbf{Scenario Two (Upper Bounds)}: We demonstrate that the e2e 
    latency of TT frames is the upper bound of the proposed architecture 
    in the presence of the uncertainty of BE transmission. So, we assign 
    copies of TT frames the same priority as the other BE data. 
    We adjust the bandwidth of the other BE data to queue or even 
    drop copies, and compare the e2e latency of SWA with that of 
    the previous TT transmission.
    \item \textbf{Scenario Three (Self-recovery)}:
    We demonstrate SWA's self recovery by forcing BE transmission 
    congestion to happen. 
    In such a case, some copies of TT frames will be queued or even dropped 
    while others that are delivered faster than TT frames, thus improving 
    TT transmission according to the self-recovery property. 
    So, we assign copies of TT frames the same priority as the other BE data, 
    and adjust BE bandwidth to queue or even drop copies. Furthermore, 
    we let the queuing and loss only happen in some switches to 
    demonstrate that the other congestion-free switches can still improve 
    latency of TT transmission. We present the dynamic improvement in 
    self-recovery scenarios by comparing the e2e latency of SWA with that 
    of the previous TT transmission.
    \item \textbf{Scenario Four (Controllable Jitter)}: 
    We show the controlled jitter of our extended architecture by a 
    configurable jitter. So, based on Scenario Two, we configure the 
    jitter of each flow according to Theorem \ref{t1}, and compare the 
    e2e latency of SWA with that of the previous TT transmission.
\end{itemize}

\begin{table}
\centering
\caption{The basic information of TT flows, namely \textit{Flow ID}, \textit{Length} and \textit{Period}.}
\label{table:flows}
\begin{tabular}{|c|c|c|c|c|c|}
 \noalign{\smallskip}
 \hline
 \noalign{\smallskip}
 \noalign{\smallskip}
 \multicolumn{2}{|c|}{Flow ID = 1} & \multicolumn{2}{|c|}{Flow ID = 2} & \multicolumn{2}{|c|}{Flow ID = 3}\\
 \noalign{\smallskip}
 \hline
 \noalign{\smallskip}
  Length  & Period  &  Length  & Period  &  Length  & Period \\
 \noalign{\smallskip}
 (Bytes)  &   (ns)  &  (Bytes) &  (ns)   & (Bytes)  &  (ns)  \\
 \noalign{\smallskip}
 \hline
 \noalign{\smallskip}
  128   & 524288  &  256   & 1048576  &  512   & 2097152\\
 \noalign{\smallskip}
 \hline
 \noalign{\smallskip}
\end{tabular}
\end{table}

\begin{table}
\centering
\caption{The configuration of TT flows per switch. The \textit{Flow-id} 
which is the same as that in Table.~\ref{table:flows} is the same flow. 
The port numbers correspond to those in Fig..~\ref{fig:platform}.}
\label{table:switches}
\begin{tabular}{|c|c|c|c|c|c|c|}
 \noalign{\smallskip}
 \hline
 \noalign{\smallskip}
 \noalign{\smallskip}
 \multirow{2}{*}{Switch} & Flow & Input & Output & Arrival   & Arrival & Offset \\
 \noalign{\smallskip}
                         & ID   & Port  & Port   & Start(ns) & End(ns) &  (ns)   \\
 \noalign{\smallskip}
 \hline
 \noalign{\smallskip}
 \multirow{3}{*}{TTS-1} & 1  & 4 & 7 & 400 & 1400 & 22528 \\ 
 \noalign{\smallskip}
 \cline{2-7}
 \noalign{\smallskip}
                        & 2  & 4 & 7 & 29072 & 30072 & 61440 \\
 \noalign{\smallskip}
 \cline{2-7}
 \noalign{\smallskip}
                        & 3  & 4 & 7 & 67984 & 68984 & 120832 \\
 \noalign{\smallskip}
 \hline
 \noalign{\smallskip}
 \multirow{3}{*}{TTS-2} & 1  & 0 & 1 & 22928 & 23928 & 45056 \\ 
 \noalign{\smallskip}
 \cline{2-7}
 \noalign{\smallskip}
                        & 2  & 0 & 1 & 61840 & 62840 & 94208 \\
 \noalign{\smallskip}
 \cline{2-7}
 \noalign{\smallskip}
                        & 3  & 0 & 1 & 121232 & 122232 & 174080 \\
 \noalign{\smallskip}
 \hline
 \noalign{\smallskip}
 \multirow{3}{*}{TTS-3} & 1  & 0 & 2 & 45456 & 46456 & 67584 \\ 
 \noalign{\smallskip}
 \cline{2-7}
 \noalign{\smallskip}
                        & 2  & 0 & 2 & 94608 & 95608 & 126976 \\
 \noalign{\smallskip}
 \cline{2-7}
 \noalign{\smallskip}
                        & 3  & 0 & 2 & 174480 & 175480 & 227328 \\
 \noalign{\smallskip}
 \hline
 \noalign{\smallskip}
\end{tabular}
\end{table}

\begin{figure*}[htbp]
  \centering
  \includegraphics[width=7.2in]{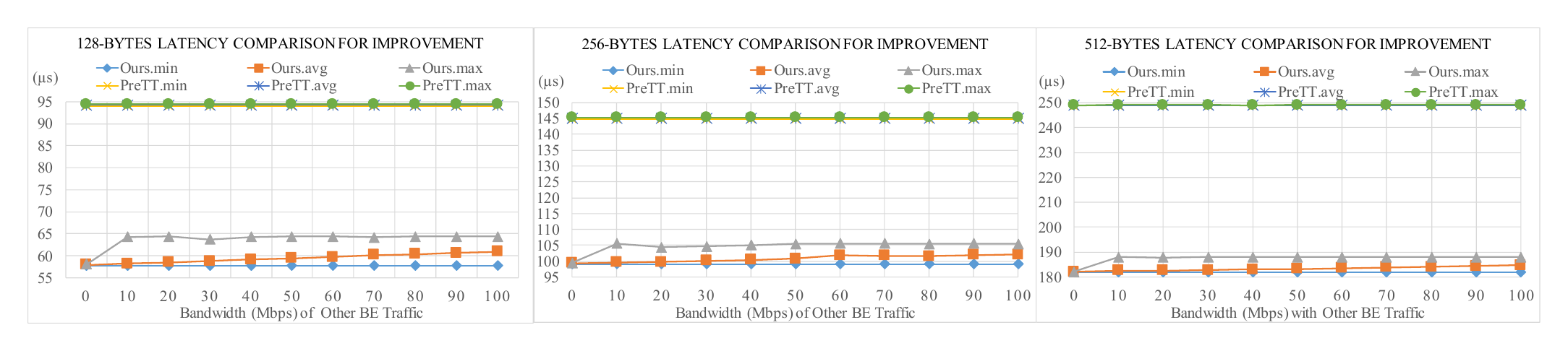}
  \caption{The e2e latency comparison of SWA and the previous TT 
  transmission, denoted by \textit{PreTT}, under Scenario One. 
  The copies of TT frames of different sizes in SWA have higher 
  priority than that of the other BE traffic which is broadcast 
  and increased by 10 Mbps as step size.}
  \label{fig:lowpriority}
\end{figure*}

\begin{figure*}[htbp]
  \centering
  \includegraphics[width=7.2in]{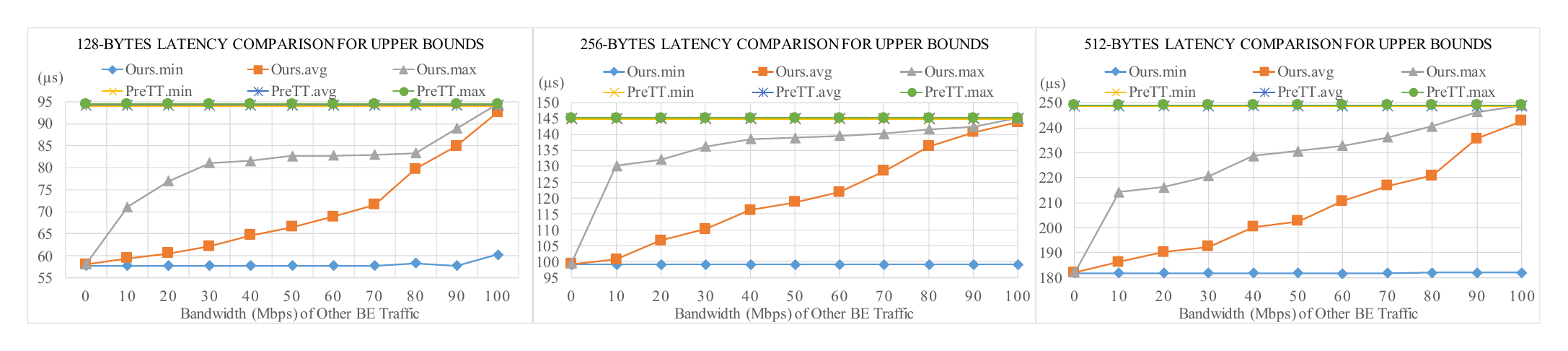}
  \caption{Comparison of e2e latencies of SWA and the previous TT 
  transmission, denoted by \textit{PreTT}, under Scenario Two. 
  The copies of TT frames of different sizes have the same priority 
  as that of the other BE traffic which is broadcast and 
  increased by 10 Mbps as step size.}
  \label{fig:samepriority}
\end{figure*}

\textbf{Test Setup:} Fig.~\ref{fig:platform} presents the experimental 
platform for the SWA. The \textit{IXIA} is a standard tester 
supporting IEEE 1588 protocol. \textit{TTS-1, TTS-2, TTS-3} 
are our TT switches. These switches establish time synchronization using 
a peer-to-peer transport clock strategy defined in the IEEE standard 1588. 
The numbers and arrowed lines highlight switch ports and links, 
respectively, which constitute the topology for data flows. 
We schedule three TT flows of different lengths and periods 
in the topology by the scheduler \cite{SteinerSMTScheduler}. 
Table~\ref{table:flows} illustrates the basic information of 
three flows: \textit{flow-id}, \textit{length}, and \textit{period}. 
Table~\ref{table:switches} provides the configuration of TT 
flows in each switches. Since Tables \textit{static-route-table}, 
\textit{sequence-table}, and \textit{filter-table} for copies 
of TT frames can be directly generated by Tables~\ref{table:flows} 
and \ref{table:switches} according to Fig.~\ref{fig:tables}, they are 
omitted here. 
The data field \textit{sequence} in these tables is initialized 
as 0 and updated based on the proposed SWA. 
All these TT flows are sent from the source, the port $0$ of \textit{IXIA} 
and are received from the sink, the port $3$ of \textit{IXIA}. 
The e2e latency of TT frames, defined as the latency from the first bit 
sent out by the source to the first bit received by the sink, is measured 
by the standard tester, \textit{IXIA}. 
The ports, $1$ and $2$ of \textit{IXIA} are used to send and receive 
other BE data as disturbance traffic, respectively, with the fixed 
length of 64 bytes and the bandwidth increased by 10 Mbps as step size. 
We test the e2e latency of TT frames with the previous TT transmission 
and SWA, respectively, under the four different scenarios as follows.

\textbf{Scenario One (Latency Improvement):} We assign the copies of TT 
frames higher priority than that of the other BE traffic, and broadcast 
the BE traffic via port 1 of \textit{IXIA}. First, we disable the functions 
of SWA by configuring the TT switches to test the e2e latency of TT frames 
with the previous TT transmission. Fig.~\ref{fig:lowpriority} plots the 
latency of TT frames of 128, 256 and 512 bytes. The minimum, average, and 
maximum latency --- \textit{PreTT.min}, \textit{PreTT.avg}, and 
\textit{PreTT.max}, respectively --- are three nearly coincident lines 
due to the high time-synchronized precision, $500 ns$ by IEEE 1588. 
The latency of TT frames of 128 bytes is about 94.25 $\mu$s and not 
affected by the broadcast of the BE traffic because of the guard band 
strategy according to IEEE 802.1 Qbv. Second, we enable the functions of 
SWA and test the e2e latency of TT frames again. Fig.~\ref{fig:lowpriority} 
illustrates the minimum, average, and maximum latency of SWA, namely, 
\textit{Ours.min}, \textit{Ours.avg}, and \textit{Ours.max}, respectively. 
Compared to the previous TT transmission, SWA makes a significant 
improvement in the latency of different frame lengths since the copies 
are forwarded as soon as possible to enhance the transmission of TT frames. 
For example, for the 128-bytes frames, the maximum latency of SWA is 
about 64.34 $\mu$s while the latency of the previous TT transmission is 
about 94.25 $\mu$s. Furthermore, the low-priority
BE traffic makes a small impact on the latency of SWA, namely, 
a 6.5 $\mu$s difference between the maximum and the minimum latency. 
Even when the bandwidth of the BE traffic reaches 100 Mbps, the 
high-priority copies are still delivered quickly to improve the latency 
of TT frames.

\begin{figure*}[htbp]
  \centering
  \includegraphics[width=7.2in]{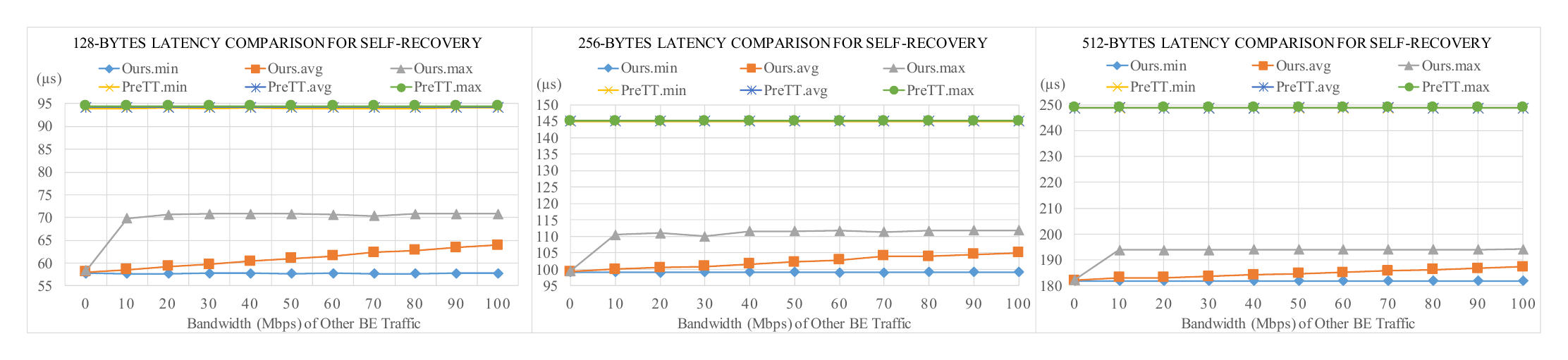}
  \caption{Comparison of e2e latencies of SWA and the previous TT 
  transmission, denoted by \textit{PreTT}, under Scenario Three. 
  The copies of TT frames of different sizes have the same priority as 
  that of the other BE traffic which is unicast and increased by 10 Mbps 
  as step size.}
  \label{fig:unicast}
\end{figure*}

\begin{figure*}[htbp]
  \centering
  \includegraphics[width=7.2in]{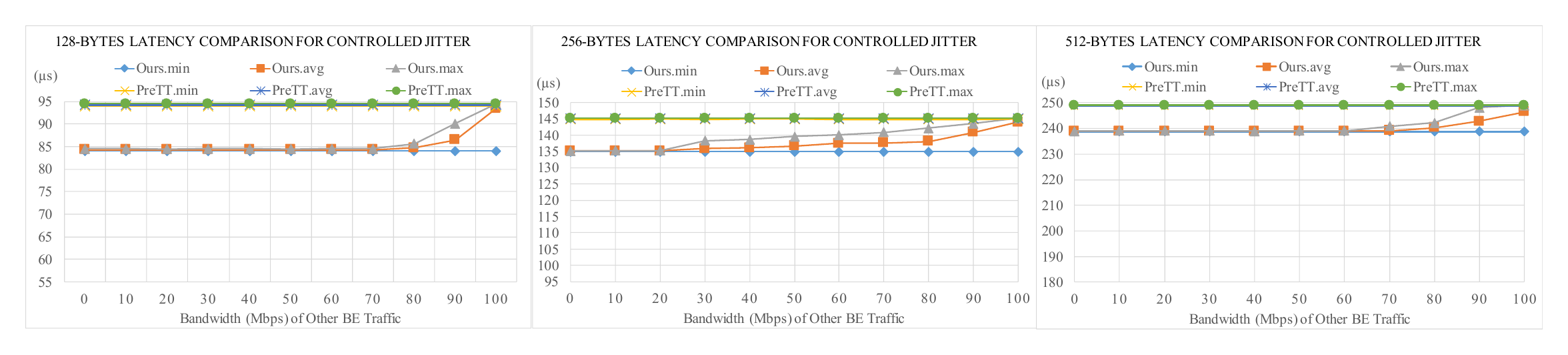}
  \caption{Comparison pf e2e latencies of SWA and the previous TT 
  transmission, denoted by \textit{PreTT}, under Scenario Four. 
  The copies of TT frames of different lengths have the same priority 
  as that of the other BE traffic which is broadcast and increased 
  by 10 Mbps as step size. These copies are transmitted with 
  configured jitters.}
  \label{fig:controledjitter}
\end{figure*}

\textbf{Scenario Two (Upper Bounds):} The difference from Scenario One is 
to broadcast the other BE traffic with the same priority as the copies 
of TT frames. First, we still disable the functions of SWA and test 
the e2e latency of TT frames with the previous TT transmission. 
As a result, Fig.~\ref{fig:samepriority} demonstrates the same latency 
of \textit{PreTT} as that in Fig.~\ref{fig:lowpriority} due to the 
scheduled precise sending instants and the guard band strategy. 
Second, we enable the proposed SWA and test the e2e latency of TT frames 
again. As illustrated in Fig.~\ref{fig:samepriority}, the latency 
changes from the minimum latency up to the latency of the previous 
TT transmission with different bandwidth of other BE traffic. 
So, the same priority BE traffic makes a bigger impact on the latency 
of SWA than that of the low-priority traffic in Scenario One. 
Furthermore, when the bandwidth of the BE traffic reaches 100 Mbps, 
the copies start to be dropped due to congestion and our latency 
reaches up to the latency of the previous TT transmission because SWA 
selects the first arrival as the final delivery. Hence, the latency of 
TT frames with the previous TT transmission is an upper bound of 
SWA's latency.

\textbf{Scenario Three (Self-recovery):} The difference from Scenario 
Two is to unicast the other BE traffic from port $1$ to port $2$ of 
\textit{IXIA}. As a result, the unicast BE traffic will disturb the 
transmission of copies of TT frames in the link from port $7$ of 
\textit{TTS-1} to port $0$ of \textit{TTS-2}, and not affect the link 
from port $1$ of \textit{TTS-2} to port $0$ of \textit{TTS-3}. 
First, we still disable the functions of SWA to test the e2e latency of 
TT frames with the previous TT transmission. Fig.~\ref{fig:unicast} 
demonstrates the same latency of \textit{PreTT} as that in 
Fig.~\ref{fig:lowpriority} and Fig.~\ref{fig:samepriority} because of 
the scheduled precise sending instants and the guard band strategy. 
Second, we enable the proposed SWA and test the e2e latency of TT frames 
again. As illustrated in Fig.~\ref{fig:unicast}, the latency of TT frames 
using SWA remains low with varying unicast BE bandwidth, which is a sharp 
contrast with the various latencies in Fig.~\ref{fig:samepriority}. 
This is because the BE traffic unicast only impacts the transmission of 
copies in \textit{TTS-1}. Even if copies are dropped in the presence
of 100 Mbps BE traffic, their transmission will recover in the next switch 
\textit{TTS-2} and continue improving the latency of the remaining path. 
So, the SWA makes a dynamic and opportunistic improvement of the latency 
of TT frames.

\textbf{Scenario Four (Controllable Jitter):} Besides the same scenario 
settings as Scenario Two, we configure the jitters of each flow to 
constrain the latency jitter of SWA. We compute a safe-jitter range for 
each flow according to Theorem \ref{t1}. 
As a result, $[0, 442496]ns$ for $flow\_id = 1$, $[0, 47232]ns$ for 
$flow\_id = 2$, and $[0, 77952]ns$ for $flow\_id = 3$. So, we set the 
jitter configuration, $jitter = 10 \mu$s for all flows and configure the 
switch \textit{TTS-3}. Fig.~\ref{fig:controledjitter} illustrates the 
latency of SWA with the configuration of the controlled jitter, i.e., 
$jitter = 10 \mu$s. Compared to the latency of SWA in Scenario Two 
illustrated in Fig.~\ref{fig:samepriority}, the lower bound of the 
latency is increased and the jitter of TT frames is controlled to be
within about 10 $\mu$s because the copies are held until the jitter is 
less than 10 $\mu$s. For example, for the TT frames of 128 bytes, the lower 
bound latency is increased to about 84.2 $\mu$s. When the bandwidth of 
disturbance traffic is below 70 Mbps, copies are held until their latency 
reaches the lower bound so that the jitter is no more than 10 $\mu$s. 
The jitter configuration constrains the lower latency bound to meet the 
jitter requirement, suggesting existence of a  tradeoff between jitter 
and latency. 
That is, the smaller the jitter, the higher the lower latency bound.

In addition, in all scenarios, we carefully compare the arrival order of 
TT frames of SWA with that of the previous TT transmission. 
The proposed SWA is order-perserving.
So, these scenarios demonstrate that the proposed SWA tackles the 
uncertainty of BE transmission and makes a dynamic improvement of the e2e 
latency of the previous TT transmission. Furthermore, to make a good 
improvement, it is preferable to assign the copies of TT frames higher 
priority than that of other BE traffic.

\section{Conclusion}
TT transmission prevents the delivery of TT frames as soon as possible 
due to the scheduled precise sending instants, while BE transmission may 
allow BE frames to be transmitted as soon as possible but may not satisfy 
the applications' e2e latency requirements. We have proposed a synergistic
switch architecture (SWA) by exploiting BE transmission to dynamically 
and opportunistically enhance TT transmission. Specifically, we have 
presented the processing steps, requisite table configuration, and 
algorithms of SWA. We have rigorously investigated the order-preservation, 
self-recovery, dynamic improvement, and cost-efficiency of the SWA. 
Furthermore, we have extended the architecture to support the configurable 
latency jitter by computing the safe jitter range for each TT flow. 
Finally, we have implemented the proposed SWA in commercial TT switches 
based on FPGAs and used four scenarios to demonstrate the SWA's capability 
of dealing with the uncertainty of BE transmission and dynamically 
improving the e2e latency of TT transmission.


%








\bibliographystyle{IEEEtran} \balance
\bibliography{reference.bib}
\end{document}